\newif\ifreport\reporttrue

\ifreport
\documentclass[journal]{IEEEtran}
\else
\documentclass[conference]{IEEEtran}
\fi
\usepackage{mathrsfs}
\usepackage{amsfonts}
\usepackage{amsthm}
\usepackage{stmaryrd}
\usepackage{bm}
\usepackage{amsmath,amssymb}
\usepackage{graphicx,cite}
\usepackage{hyperref}
\usepackage[tight,footnotesize]{subfigure}
\newtheorem{lemma}{Lemma}

\newtheorem{Theorem}{Theorem}

\newcommand\xb{\ensuremath{{\bm x}}}

\newcommand\yb{\ensuremath{{\bm y}}}

\newcommand\Ab{\ensuremath{{\bf A}}}

\newcommand\Bb{\ensuremath{{\bf B}}}
\newcommand\Cb{\ensuremath{{\bf C}}}

\newcommand\Db{\ensuremath{{\bf D}}}

\newcommand\Eb{\ensuremath{{\bf E}}}

\newcommand\Gb{\ensuremath{{\bf G}}}

\newcommand\Hb{\ensuremath{{\bf H}}}

\newcommand\Ib{\ensuremath{{\bf I}}}

\newcommand\Mb{\ensuremath{{\bf M}}}
\newcommand\Nb{\ensuremath{{\bf N}}}
\newcommand\nb{\ensuremath{{\bf n}}}
\newcommand\Pb{\ensuremath{{\bf P}}}
\newcommand\Qb{\ensuremath{{\bf Q}}}

\newcommand\Vb{\ensuremath{{\bf V}}}

\newcommand\Xb{\ensuremath{{\bf X}}}

\newcommand\Ub{\ensuremath{{\bf U}}}

\newcommand\Deltab{\ensuremath{{\bm \Delta}}}

\newcommand\Phib{\ensuremath{{\bm \Phi}}}

\newcommand\tr{\ensuremath{{\rm Tr}}}

\newcommand\Det{\ensuremath{{\rm det}}}

\newcommand\Cplx{\ensuremath{{\mathbb{C}}}}
\newcommand\Real{\ensuremath{{\mathbb{R}}}}
\newcommand\Hmit{\ensuremath{{\mathbb{S}}}}

\makeatletter
\renewcommand{\maketag@@@}[1]{\hbox{\m@th\normalsize\normalfont#1}}%
\makeatother

\usepackage{graphics,booktabs,color,epsfig,subfigure}

\begin{document}

\bibliographystyle{IEEEtran}

\title{Capacity of Compound MIMO Gaussian Channels with Additive Uncertainty}

\author{Yin Sun, C. Emre Koksal, and Ness B. Shroff\\
Dept. of ECE, The Ohio State University, Columbus, OH\\
\thanks{This paper will be presented in part at IEEE ISIT 2013.}
}

\maketitle
\makeatletter
    \def\@eqnnum{{\normalsize \normalcolor (\theequation)}}
\makeatother
\begin{abstract}
This paper considers reliable communications over a multiple-input multiple-output (MIMO) Gaussian channel, where the channel matrix is within a bounded channel uncertainty region around a nominal channel matrix, i.e., an instance of the compound MIMO Gaussian channel. We study the {optimal transmit covariance matrix design} to achieve the capacity of compound MIMO Gaussian channels, where the channel uncertainty region is characterized by the spectral norm. This design problem is a challenging non-convex optimization problem. However, in this paper, we reveal that this problem has a \emph{hidden convexity} property, {which can be exploited to map the problem into} a convex optimization problem. We first prove that the optimal transmit design is to diagonalize the nominal channel, and then show that the duality gap between the capacity of the compound MIMO Gaussian channel and the {min-max channel capacity} is zero, which proves the conjecture of Loyka and Charalambous (\emph{IEEE Trans. Inf. Theory}, vol. 58, no. 4, pp. 2048-2063, 2012). The key tools for showing these results are a new matrix determinant inequality and some unitarily invariant properties.
\end{abstract}

\IEEEpeerreviewmaketitle

\section{Introduction}
Multiple-input multiple-output (MIMO) techniques have been {extensively used}
to improve the spectral efficiencies of wireless communications. The performance of MIMO communications relies on access to the channel state information (CSI). When the CSI is perfectly known at the transmitter, the optimal power allocation is to diagonalize the channel \cite{Telatar1999}. However, in practice, the transmitter often has some channel uncertainty, which can result in a significant rate loss, if not taken into consideration in the transmit covariance matrix design.

{There have been two categories of research} towards reliable communications over MIMO Gaussian channels with channel uncertainty. The first category focuses on stochastic models of channel uncertainty, where the transmitter has access to only the statistics of the channel state, but not its realization. When the channel states change quickly over time, the achievable rate of the channel is described by the ergodic capacity, e.g., \cite{Telatar1999,Foschini1998,Moustakas2003,Simon2003}. On the other hand, when the channel states vary slowly, the achievable rate is characterized by the outage capacity, which is the maximum data rate achievable at any given state with probability no smaller than a specified value,
e.g., \cite{Telatar1999,Moustakas2003,Simon2003,Jafar2004,Moustakas2007,Kazakopoulos2011}. 

The second category of studies were centered on deterministic models of channel uncertainty, where the CSI is a deterministic variable within a known set, but its actual value is unknown to the transmitter. Such a model is called a \emph{compound channel} in information theory, {and its capacity is determined as the maximum of the worst-case mutual information (max-min channel capacity) of the set of possible channel realizations \cite{Blackwell1959}.} From practical viewpoint, it is the maximum data rate that can be reliably transmitted over \emph{any} channel from the given set. Characterizing the capacity of the compound channel is considered to be an important problem, and has received considerable attention.

 {In closed-loop MIMO systems, the transmitter is able to obtain inaccurate CSI, where the channel error may be caused by estimation, interpolation, mobility, and/or feedback.} In this case, the channel is typically modelled as the sum of a known nominal channel and unknown channel uncertainty. This additive channel uncertainty model has been widely utilized {in studies on the fundamental limits of MIMO channels, e.g., \cite{Visotsky2001,Goldsmith2003,Loyka2012}, and on robust transceiver designs}, e.g., \cite{Vorobyov2003,Lorenz2005,Wang2009,Huang2012}.
In \cite{Wiesel2007}, the capacity of the compound Rician MIMO Gaussian channel with additive channel uncertainty was studied, where the analysis was restricted to a rank-one nominal channel. Arbitrary rank nominal channel was considered in \cite{Denic2009}, where the channel uncertainty is limited to the singular value of the nominal channel with no uncertainty on the singular vectors.
The capacity of the compound MIMO channel with a multiplicative channel uncertainty model was obtained in \cite{Loyka2012}, where the region of channel uncertainty is described by spectral norm. In addition, the capacity of the compound MIMO Gaussian channel with additive channel uncertainty was derived in \cite{Loyka2012} for some special cases, such as high signal-to-noise ratio (SNR) limit, low SNR limit, and rank-two nominal channel. 


In this paper, we design the optimal transmit covariance matrix to achieve the capacity of the compound MIMO Gaussian channel with additive channel uncertainty.
We consider the case {where} the channel uncertainty is in a bounded region around the nominal channel matrix, which is characterized by the spectral norm.
This design problem is a challenging nonconvex optimization problem. However, we reveal that this problem possesses a \emph{hidden convexity} property, and {hence can be simplified into} a convex optimization problem. We first prove that the optimal transmit design is to diagonalize the nominal channel. {We then} show that the duality gap between the capacity of the compound MIMO Gaussian channel (max-min channel capacity) and the {min-max channel capacity} is zero, which proves the conjecture of Loyka and Charalambous \cite{Loyka2012}.
The key tools for proving these results are a new matrix determinant inequality { (Lemma~\ref{lem1})} and some unitarily invariant properties.

\section{System Model}
\label{sec:model}
\subsection{Notation}\label{sec:notation}
The following notations are used {throughout the paper}. Boldface upper-case letters denote matrices, boldface lower-case letters denote column vectors, and standard lower case letters denote scalars. Let $\Cplx^{m\times n}$ denote the set of $m\times n$ complex-valued matrices, and $\Cplx^{n}$ denote the set of $n\times n$ square complex-valued matrices. The symbol $\Hmit^n$ represents the set of $n\times n$ Hermitian matrices, and
$\Hmit_+^n$ represents the set of $n\times n$ Hermitian positive semidefinite matrices.
\ifreport
Let $\Xb(S)$ denote a submatrix of $\Xb$ obtained by deleting the rows and columns complementary to those indicated by $S$ from $\Xb$.
\else
\fi
The operator $\textrm{diag}(x_1,x_2,\cdots,x_n)$ denotes a diagonal matrix with diagonal entries given by $x_1,x_2,\cdots,x_n$. The matrix $\Ib_{n}$ denotes the $n\times n$ identity matrix.
By $\xb\geq\bm 0$, we mean that $x_i\geq0$ for all $i$.
The operators $(\cdot)^H$, $\tr(\cdot)$, and $\Det(\cdot)$ on matrices denote the Hermitian, trace, and determinant operations, respectively. Let $\sigma_i(\Ab)$ and $\lambda_i(\Ab)$ represent the singular value and eigenvalue of $\Ab$, respectively.
The vector $\bm \sigma(\Ab)\triangleq (\sigma_1(\Ab),\cdots, \sigma_{\min\{m,n\}}(\Ab))$ contains the singular values of $\Ab\in\Cplx^{m\times n}$. Let
$\bm \lambda(\Qb)\triangleq (\lambda_1(\Qb),\cdots, \lambda_{n}(\Qb))$ denote a vector containing the eigenvalues of $\Qb\in\Hmit^n$.
The singular values and eigenvalues are listed in descending order. We use $\interleave\cdot\interleave$ and $\|\cdot\|$ to denote matrix norm and vector norm, respectively.
\subsection{Channel Model}
Consider the complex-valued Gaussian vector channel:
\begin{eqnarray}\label{eq:channel}
\yb =\Hb\xb+\nb,
\end{eqnarray}
where $\yb$ is a length $r$ received vector, $\Hb$ is an $r\times t$ channel matrix, $\xb$ is a length $t$ transmitted vector with zero mean and covariance $E\{\xb\xb^H\}=\Qb$, and $\nb$ is a complex Gaussian noise vector with zero-mean and covariance $E\{\nb\nb^H\}=\Ib_r$.

The MIMO channel $\Hb$ is an unknown deterministic matrix satisfying
\begin{eqnarray}\label{eq:uncertainty}
\Hb\in \mathcal{H},
\end{eqnarray}
where $\mathcal{H}$ is the channel uncertainty region defined by
\begin{eqnarray}\label{eq:channel_con}
\mathcal{H}\triangleq\{\Hb:\interleave\Hb-\Hb_0\interleave_2\leq \varepsilon\},
\end{eqnarray}
$\Hb_0$ is the nominal channel, and $\interleave\cdot\interleave_2$ is the spectral norm defined by
\begin{eqnarray}\label{eq:spectral_norm}
\interleave\Ab\interleave_2\triangleq\max_{\|\xb\|_2\leq1}\|\Ab\xb\|_2=\max_i\{\sigma_{i}(\Ab)\}=\|\bm\sigma(\Ab)\|_\infty.
\end{eqnarray}

The spectral norm is a \emph{unitarily invariant matrix norm}. A unitarily invariant matrix norm satisfies \cite[Section 7.4.16]{Hornbook1985}
\begin{eqnarray}\label{eq:channel_con_property}
\interleave\Ub\Ab\Vb\interleave= \interleave\Ab\interleave
\end{eqnarray}
for all $\Ab\in \Cplx^{m\times n}$ and for all unitary matrices $\Ub\in\Cplx^{m}$ and $\Vb\in\Cplx^{n}$. Therefore, the channel uncertainty $\bm\Delta=\Hb-\Hb_0$ is within an isotopical set.
\ifreport
\footnote{This is different from the channel uncertainty model in \cite{Palomar2003}, where $\Hb$ is within an isotopical set such that $\Hb\Ub\in \mathcal{H}$ for all $\Hb\in \mathcal{H}$ and all unitary matrix $\Ub$.}
\else
\fi
Note that the channel uncertainty region \eqref{eq:channel_con} provides a conservative performance lower bound for the regions defined by any other unitarily invariant matrix norm, because
$$\interleave\Ab\interleave_2\geq \interleave\Ab\interleave$$
holds for all matrix $\Ab$ and all unitarily invariant matrix norm $\interleave\cdot\interleave$ \cite[Corollary 5.6.35]{Hornbook1985}.
\ifreport
More discussions {on the relationship among some matrix norms} are provided in Section \ref{sec:discussion}.
\else
More discussions on the relationship among some matrix norms are provided in \cite{report_MIMO_compound2013}.
\fi

\subsection{Power Constraint}
We consider a general transmit power constraint
\begin{eqnarray}\label{eq:power-constraint}
\Qb\in\mathcal {Q},
\end{eqnarray}
where $\mathcal{Q}\subset\Hmit_+^t$ is a nonempty compact convex set satisfying
\begin{eqnarray} \label{eq:power-region_property}
&&\Ub\Qb\Ub^H\in \mathcal {Q},\\
&&\Db(\Qb)\in \mathcal {Q}, \label{eq:power-region_property1}
\end{eqnarray}
for all $\Qb\in \mathcal {Q}$ and all unitary matrix $\Ub\in\Cplx^{t}$, where $\Db(\Qb)$ is the diagonal matrix with the same diagonal elements with $\Qb$. We say {that} a set $\mathcal {Q}$ is \emph{unitarily invariant} if it satisfies \eqref{eq:power-region_property} and \eqref{eq:power-region_property1}.
One can show that each unitarily invariant $\mathcal {Q}$ can be equivalently expressed as
\begin{eqnarray}\label{eq:convex_set}
\mathcal {Q}=\{\Qb\in \Hmit_+^t:\bm \lambda(\Qb)\in B_\mathcal {Q}, \bm \lambda(\Qb)\geq\bm 0\},
\end{eqnarray}
where $B_\mathcal {Q}$ is a nonempty compact convex set.
Two {typical} examples of unitarily invariant power constraints are the sum power constraint \cite{Telatar1999}
\begin{eqnarray}\label{eq:power-constraint1}
\mathcal{Q}_1\!\!\!\!\!\!\!\!\!\!\!&&=\{\Qb\in \Hmit_+^t:\tr(\Qb)\leq t\}, \\
&&=\{\Qb\in \Hmit_+^t:\sum_{i=1}^t \lambda_i(\Qb)\leq t, \bm \lambda(\Qb)\geq\bm 0\},\nonumber
\end{eqnarray}
and the maximum power constraint \cite{Wang2009}
\begin{eqnarray}
\mathcal{Q}_2=\{\Qb\in \Hmit_+^t:\max_i\{\lambda_{i}(\Qb)\}\leq P_{m}, \bm \lambda(\Qb)\geq\bm 0\}. \nonumber
\end{eqnarray}
\section{Optimal Transmit Covariance Design}
\subsection{Main Result}
The capacity of the compound MIMO Gaussian channel \eqref{eq:channel}-\eqref{eq:channel_con} and \eqref{eq:power-constraint} is \cite[Theorem 7.1]{Gamalbook2011}
\begin{eqnarray}\label{eq:maxmin}
C_{\max\min}\triangleq\max_{\Qb\in\mathcal{Q}} \min_{\interleave\Hb-\Hb_0\interleave_2\leq \varepsilon} I(\Qb,\Hb),
\end{eqnarray}
where $I(\Qb,\Hb)=I(\xb;\yb)$ is the mutual information of the channel \eqref{eq:channel}, i.e., \cite{Telatar1999}
$$I(\Qb,\Hb)=\log\Det\left(\Ib_r + \gamma\Hb\Qb\Hb^H\right),$$
and $\gamma$ is the per-antenna SNR. {Finding an efficient solution} of the max-min problem \eqref{eq:maxmin} has been open for a long time (except {in} some special cases \cite{Wang2009,Wiesel2007,Loyka2012}), because $I(\Qb,\Hb)$ is nonconvex with respect to $\Hb$.
\ifreport
{However, we} will show that the problem \eqref{eq:maxmin} possesses a hidden convexity property when \eqref{eq:channel_con} holds, and thus can be simplified {into} a convex optimization problem.
\else
\fi

Suppose that the singular value decomposition (SVD) of the nominal channel $\Hb_0$ is given by
\begin{eqnarray}\label{eq:SVD_H0}
\Hb_0=\Ub_0 \bm\Sigma_{\Hb_0} \Vb_0^H,
\end{eqnarray}
where $\Ub_0\in\Cplx^{r}$ and $\Vb_0\in\Cplx^{t}$ are unitary matrices.
The first key result of this paper is stated as follows:

\begin{Theorem}\label{thm1}
If $\mathcal{Q}$ and $\mathcal{H}$ are nonempty sets, $\mathcal {H}$ is defined in \eqref{eq:channel_con}, and $\mathcal {Q}$ satisfies the unitarily invariant properties \eqref{eq:power-region_property} and \eqref{eq:power-region_property1}, then
\begin{eqnarray}\label{eq:solution}
\Qb^\star = \Vb_0 \bm\Lambda_\Qb^\star \Vb_0^H,~\Hb^\star = \Ub_0 \bm\Sigma_\Hb^\star \Vb_0^H,
\end{eqnarray}
is a solution to Problem \eqref{eq:maxmin},
where $\Ub_0$ and $\Vb_0$ are defined in \eqref{eq:SVD_H0}, the diagonal matrices $\bm\Lambda_\Qb^\star$ and $\bm\Sigma_\Hb^\star$ are determined by $\bm\Sigma_\Hb^\star = \textrm{diag} (\bm \sigma^\star)$ and
$\bm\Lambda_\Qb^\star = \textrm{diag} (\bm \lambda^\star)$, such that $(\bm \sigma^\star,\bm \lambda^\star)$ solves the problem
\begin{eqnarray}\label{eq:vec_opt_problem}
C_{\max\min}=\max_{\substack{\bm \lambda\in B_\mathcal {Q}\\\bm \lambda\geq\bm0}}\min_{\substack{\|\bm \sigma-\bm \sigma_0\|_\infty\leq\varepsilon\\\bm \sigma\geq\bm0}}\sum_{i=1}^{\min\{t,r\}}\log (1+\gamma\sigma_i^2\lambda_i),
\end{eqnarray}
with the convex set $B_\mathcal {Q}$ defined in \eqref{eq:convex_set}.
\end{Theorem}
\begin{proof}
The proof of Theorem \ref{thm1} relies on the unitarily invariant properties { \eqref{eq:spectral_norm},} \eqref{eq:channel_con_property}, \eqref{eq:power-region_property}, and \eqref{eq:power-region_property1}, and a new matrix determinant inequality presented in Lemma~\ref{lem1} given below. {The details of the proof} are provided in Appendix \ref{app1}.
\end{proof}
{The following lemma is a {key} technical contribution of this paper, which plays {an important} role in proving Theorem \ref{thm1}.}

\begin{lemma}[Matrix Determinant Inequality]\label{lem1}
If $\bm\Sigma$ and $\bm\Lambda$ are diagonal matrices with nonnegative diagonal entries, then one solution to
\begin{eqnarray}\label{eq:lem1}
\min_{\interleave{\bm\Delta}\interleave_2\leq \varepsilon} \Det\left[\Ib\! + \! (\bm\Sigma\!+\!{\bm\Delta})\bm\Lambda
(\bm\Sigma\!+\!{\bm\Delta})^H\right]
\end{eqnarray}
is a diagonal matrix.
\end{lemma}
\ifreport
\begin{proof}
See Appendix \ref{app3}.
\end{proof}
\else
The proof of Lemma~\ref{lem1} is provided in \cite{report_MIMO_compound2013}.
\fi

Theorem \ref{thm1} implies that the optimal transmit covariance of the MIMO Gaussian channel with worst case channel uncertainty is to diagonalize the nominal channel $\Hb_0$. Such a solution structure was previously known only for some special cases, such as high SNR limit ($\gamma\gg 1$), low SNR limit ($\gamma\ll 1$), low rank nominal channels (rank$(\Hb_0)\leq 2$) \cite{Wiesel2007,Loyka2012,Wang2009}. {In contrast,} Theorem \ref{thm1} holds for general nominal channels and all SNR values.
{Further, by} Theorem \ref{thm1}, the problem \eqref{eq:maxmin} reduces to \eqref{eq:vec_opt_problem} with much fewer variables.

\subsection{The Dual Problem}
Now, we consider the dual of the max-min problem \eqref{eq:maxmin}, {which is given by the following  \emph{min-max channel capacity}} problem
\begin{eqnarray}\label{eq:minmax1}
C_{\min\max}\triangleq \min_{\Hb\in \mathcal{H}}\max_{\Qb\in\mathcal{Q}} I(\Qb,\Hb).
\end{eqnarray}
It is important to distinguish the capacity of the compound channel $C_{\max\min}$ and the {min-max channel capacity} $C_{\min\max}$: $C_{\max\min}$ can be achieved for any channel $\Hb$ within $\mathcal{H}$, by using the same transmit covariance matrix $\Qb$.\footnote{The outer optimization of $\Qb$ in \eqref{eq:maxmin} is done without knowledge of $\Hb$.} $C_{\min\max}$ is the minimal capacity of the channels with $\Hb\in \mathcal{H}$, evaluating which requires
knowledge of $\Hb$ at the transmitter to obtain $\Qb$.\footnote{The inner optimization of $\Qb$ in \eqref{eq:minmax1} is done with knowledge of $\Hb$.}
%
We study the min-max problem \eqref{eq:minmax1} to gain more insight into the max-min problem \eqref{eq:maxmin}.
We consider a more general channel uncertainty region
\begin{eqnarray}\label{eq:channel_con1}
\mathcal{H}\triangleq\{\Hb:\interleave\Hb-\Hb_0\interleave\leq \varepsilon\},
\end{eqnarray}
where $\interleave\cdot\interleave$ is a unitarily invariant matrix norm satisfying \eqref{eq:channel_con_property}.
For any unitarily invariant matrix norm $\interleave\cdot\interleave$, there is a vector norm $\|\cdot\|$ such that
\begin{eqnarray}\label{eq:vec_norm}
\interleave\Ab \interleave=\|\bm \sigma(\Ab)\|
\end{eqnarray}
holds for all $\Ab\in \Cplx^{m\times n}$ \cite[Theorem 3.5.18]{Hornbook1991}. {For the special case of spectral norm, the associated vector norm in \eqref{eq:vec_norm} is $\|\cdot\|_\infty$, as given by \eqref{eq:spectral_norm}.
We have the following result:}
\begin{Theorem}\label{thm2}
If $\mathcal{Q}$ and $\mathcal{H}$ are nonempty sets, $\mathcal {H}$ is defined in \eqref{eq:channel_con1}, and $\mathcal {Q}$ satisfies the unitarily invariant properties \eqref{eq:power-region_property} and \eqref{eq:power-region_property1}, then
\begin{eqnarray}\label{eq:thm2_1}
\Qb'= \Vb_0 \bm\Lambda_\Qb' \Vb_0^H,~\Hb' = \Ub_0 \bm\Sigma_\Hb' \Vb_0^H,
\end{eqnarray}
is a solution to Problem \eqref{eq:minmax1},
where $\Ub_0$ and $\Vb_0$ are defined in \eqref{eq:SVD_H0}, the diagonal matrices $\bm\Lambda_\Qb'$ and $\bm\Sigma_\Hb'$ are determined by $\bm\Sigma_\Hb' = \textrm{diag} (\bm \sigma')$ and
$\bm\Lambda_\Qb' = \textrm{diag} (\bm \lambda')$ such that $(\bm \sigma',\bm \lambda')$ solves the problem
\begin{eqnarray}\label{eq:vec_opt_problem1}
C_{\min\max}=\min_{\substack{\|\bm \sigma-\bm \sigma_0\|\leq\varepsilon\\\bm \sigma\geq\bm0}}\max_{\substack{\bm \lambda\in B_\mathcal {Q}\\\bm \lambda\geq\bm0}}\sum_{i=1}^{\min\{t,r\}}\log (1+\gamma\sigma_i^2\lambda_i),
\end{eqnarray}
with the vector norm $\|\cdot\|$ and the convex set $B_\mathcal {Q}$ defined in \eqref{eq:vec_norm} and \eqref{eq:convex_set}, respectively.
\end{Theorem}

\begin{proof}
The proof of Theorem \ref{thm2} relies on the unitarily invariant properties { \eqref{eq:vec_norm},} \eqref{eq:channel_con_property}, \eqref{eq:power-region_property}, and \eqref{eq:power-region_property1}, but not the matrix determinant inequality in Lemma~\ref{lem1} for the spectral norm case. Therefore, Theorem \ref{thm2} holds for any unitarily invariant matrix norm.
The details of the proof are provided in Appendix \ref{app2}.
\end{proof}
{Note that a special case of Theorem \ref{thm2} was obtained in Theorem 3 of \cite{Loyka2012}, where $\interleave\cdot\interleave$ is limited to the spectral norm $\interleave\cdot\interleave_2$ and $\mathcal{Q}$ is the sum power constraint $\mathcal{Q}_1$.}

\subsection{Duality Gap is Zero}

It is interesting to see that the max-min problem \eqref{eq:maxmin} and the min-max problem \eqref{eq:minmax1} have similar solution structures, as given in \eqref{eq:solution} and \eqref{eq:thm2_1}, and the difference is only in the solutions to \eqref{eq:vec_opt_problem} and \eqref{eq:vec_opt_problem1}.
Next, we study whether \eqref{eq:vec_opt_problem} and \eqref{eq:vec_opt_problem1} have a common solution for the spectral norm case.

It is known that the following weak duality relation is always true: \cite{Gamalbook2011}
\begin{eqnarray}\label{eq:dual-gap}
C_{\max\min}\leq C_{\min\max}.
\end{eqnarray}
Moreover, equality holds in \eqref{eq:dual-gap}, i.e.,
\begin{eqnarray}\label{eq:saddle0}
C_{\max\min}=C_{\min\max},
\end{eqnarray}
if and only if \eqref{eq:vec_opt_problem} and \eqref{eq:vec_opt_problem1} have a common solution \cite[Corollary 9.16]{Zeidlerbook1986}.
It was conjectured in \cite{Loyka2012} that \eqref{eq:saddle0} holds for the case that $\interleave\cdot\interleave=\interleave\cdot\interleave_2$ and the power constraint is $\mathcal{Q}=\mathcal{Q}_1$. 
Here, using {Theorems \ref{thm1}, \ref{thm2}, and von Neumann's minimax theorem}, we can now prove this conjecture:

\begin{Theorem}\label{thm3}
If the conditions of Theorem \ref{thm1} are satisfied, then:
\begin{itemize}
\item[1)] The strong duality relation \eqref{eq:saddle0} holds.

\item[2)] Problems \eqref{eq:vec_opt_problem} and \eqref{eq:vec_opt_problem1} have a common solution $(\bm \sigma^*,\bm \lambda^*)$, where $\bm \sigma^*$ is given by
\begin{eqnarray}\label{eq:channel_solution}
\sigma_{i}^*=\max\{\sigma_{0,i}-\varepsilon,0\},
\end{eqnarray}
\end{itemize}
and $\bm \lambda^*$ is determined by the convex optimization problem
\begin{eqnarray}\label{eq:power_solution}
C_{\max\min}
=\max_{\substack{\bm \lambda\in B_\mathcal {Q}\\\bm \lambda\geq\bm0}}\!\!\!\!\sum_{i=1}^{\min\{t,r\}}\!\!\!\!\log (1\!+\!\gamma\max\{\sigma_{0,i}\!-\!\varepsilon,0\}^2\lambda_i).
\end{eqnarray}
\end{Theorem}
\begin{proof}1) Problem \eqref{eq:vec_opt_problem} can be expressed as
\begin{eqnarray}\label{eq:1}
C_{\min\max}=\max_{\substack{\bm \lambda\in B_\mathcal {Q}\\\bm \lambda\geq\bm0}}\min_{\substack{\bm \sigma\geq\bm0}}\!\!\!\!\!\!\!\!&&\sum_{i=1}^{\min\{t,r\}}\!\!\log (1+\gamma\sigma_i^2\lambda_i)\nonumber\\
\textrm{s.t.}~\!\!\!\!\!\!\!\! &&\max\{\sigma_{0,i}\!-\!\varepsilon,0\}\leq \sigma_i\leq\sigma_{0,i}\!+\!\varepsilon,~\forall~i.\nonumber
\end{eqnarray}
By introducing $x_i\triangleq\log(\sigma_i)$, this problem can be reformulated as the following convex optimization problem:
\begin{eqnarray}\label{eq1}
\max_{\substack{\bm \lambda\in B_\mathcal {Q}\\\bm \lambda\geq\bm0}}\min_{\substack{\bm x}}\!\!\!\!\!\!\!\!&&\sum_{i=1}^{\min\{t,r\}}\log [1+\gamma e^{2x_i}\lambda_i]\\
\textrm{s.t.}~\!\!\!\!\!\!\!\! &&\log(\max\{\sigma_{0,i}-\varepsilon,0\})\leq x_i\leq\log(\sigma_{0,i}+\varepsilon),~\forall~i,\nonumber
\end{eqnarray}
where the objective function is concave in $\bm \lambda$ and convex in $\bm x$ \cite{Boyd04}. {Similarly, \eqref{eq:vec_opt_problem1} can be reformulated as the following convex optimization problem:
\begin{eqnarray}\label{eq2}
\min_{\substack{\bm x}}\!\!\!\!\!\!\!\!&&\max_{\substack{\bm \lambda\in B_\mathcal {Q}\\\bm \lambda\geq\bm0}}\sum_{i=1}^{\min\{t,r\}}\log [1+\gamma e^{2x_i}\lambda_i]\\
\textrm{s.t.}~\!\!\!\!\!\!\!\! &&\log(\max\{\sigma_{0,i}-\varepsilon,0\})\leq x_i\leq\log(\sigma_{0,i}+\varepsilon),~\forall~i.
\nonumber
\end{eqnarray}
Let us use $f(\bm \lambda,\bm x)$ to denote the objective function in \eqref{eq1} and \eqref{eq2}. We say a point $(\bm \lambda_0,\bm x_0)$ is a \emph{saddle point} of $f$ if
\begin{eqnarray}\label{eq3}
f(\bm \lambda_0,\bm x_0)\!=\!\min_{\substack{\sigma_{0,i}-\varepsilon\leq e^{x_i}
\\e^{x_i}\leq\sigma_{0,i}+\varepsilon}} f(\bm \lambda_0,\bm x)= \max_{\substack{\bm \lambda\in B_\mathcal {Q}\\\bm \lambda\geq\bm0}}f(\bm \lambda,\bm x_0).
\end{eqnarray}
From von Neumann's minimax theorem \cite[Theorem 9.D]{Zeidlerbook1986}, we have:
(1) the function $f$ has a saddle point; (2) the point $(\bm \lambda_0,\bm x_0)$ is a saddle point of $f$ \emph{if and only if} $(\bm \lambda_0,\bm x_0)$ is a common solution to \eqref{eq1} and \eqref{eq2}. Therefore, \eqref{eq1} and \eqref{eq2} have the same optimal value, and \eqref{eq:saddle0} follows.

2) Let us define $\bm x^*$ as $x_{i}^*=\log (\max\{\sigma_{0,i}-\varepsilon,0\}).$
By von Neumann's minimax theorem, we only need to show that $(\bm x^*,\bm \lambda^*)$ is a saddle point of $f$.
When $(\bm \lambda_0,\bm x_0)$ is replaced by $(\bm x^*,\bm \lambda^*)$, the minimization problem in \eqref{eq3} can be separated into several subproblems, i.e.,
\begin{eqnarray}
\min_{\bm x}\!\!\!\!\!\!\!\!&&\log (1+\gamma e^{2x_i}\lambda_i^*)\nonumber\\
\textrm{s.t.}~\!\!\!\!\!\!\!\! &&\log(\max\{\sigma_{0,i}-\varepsilon,0\})\leq x_i\leq\log(\sigma_{0,i}+\varepsilon),~\forall~i,\nonumber
\end{eqnarray}
and the solution is given by $\bm x^*$. On the other hand, according to \eqref{eq:power_solution}, $\bm \lambda^*$ is the solution to the maximization problem in \eqref{eq3}. Therefore, \eqref{eq3} holds at $(\bm x^*,\bm \lambda^*)$, and thus $(\bm x^*,\bm \lambda^*)$ is a saddle point of $f$. }
\end{proof}
We note that the conjecture of \cite{Loyka2012} is a special case of Theorem \ref{thm3} where $\mathcal{Q}$ is restricted to be the sum power constraint $\mathcal{Q}_1$.
By Theorem \ref{thm1} and \ref{thm3}, we have shown that the covariance design problem \eqref{eq:maxmin} is a convex optimization problem in nature, if the channel uncertainty region $\mathcal {H}$ is characterized by the spectral norm.

\ifreport
{\section{Discussion}\label{sec:discussion}
{It is known that any matrix norm can be bounded within {a constant multiple} of the spectral norm: for any matrix norm $\interleave\cdot\interleave$, there exist $\alpha_l,\alpha_h>0$ such that
\begin{eqnarray}
\alpha_l\interleave\Hb-\Hb_0\interleave_2\leq \interleave\Hb-\Hb_0\interleave\leq \alpha_h\interleave\Hb-\Hb_0\interleave_2,
\end{eqnarray}
where $\alpha_l$ and $\alpha_h$ are independent of the channel error $\Hb-\Hb_0$ but depend on $\interleave\cdot\interleave$ \cite{Hornbook1985}. The values of $\alpha_l$ and $\alpha_h$ are summarized in Table I of \cite{Loyka2012} for some popular matrix norms. Using this and our results in Theorem 1 and 3, one can derive capacity lower and upper bounds for compound MIMO Gaussian channels with channel uncertainty region described by different matrix norms.}

Another interesting open problem is whether the matrix determinant inequality in Lemma~\ref{lem1} holds for some other unitarily invariant matrix norms, e.g., the Frobenius (Euclidean) norm
\begin{eqnarray}\label{eq:F_norm}
\interleave\Ab\interleave_F\triangleq\sqrt{\sum_{i,j} |a_{ij}|^2}=\|\bm\sigma(\Ab)\|_2.
\end{eqnarray}
If one can generalize Lemma~\ref{lem1} to another unitarily invariant matrix norm, then Theorems \ref{thm1} and \ref{thm3} also hold for this particular matrix norm, with the infinite norm $\|\cdot\|_\infty$ in \eqref{eq:vec_opt_problem} replaced by the corresponding vector norm defined by \eqref{eq:vec_norm}.

We have found a counterexample, {which shows} that Lemma~\ref{lem1} does not hold for the matrix norm $\interleave\cdot\interleave = \|\bm\sigma(\cdot)\|_1$. Consider the matrices
\begin{eqnarray}\label{eq:example2}
\bm\Sigma = \left[
          \begin{array}{cc}
            2 & 0 \\
            0 & 1 \\
          \end{array}
        \right],
        \bm\Lambda = \left[
          \begin{array}{cc}
            4 & 0 \\
            0 & 3 \\
          \end{array}
        \right].
\end{eqnarray}
If $\bm\Delta$ is restricted as a diagonal matrix, one can numerically calculate that the optimal objective value of the problem
\begin{eqnarray}\label{eq:example1}
\min_{\|\bm\sigma(\bm\Delta)\|_1\leq 1} \Det\left[\Ib\! + \! (\bm\Sigma\!+\!{\bm\Delta})\bm\Lambda
(\bm\Sigma\!+\!{\bm\Delta})^H\right]
\end{eqnarray}
is 15.63. However, when
\begin{eqnarray}
\bm\Delta = \left[
          \begin{array}{cc}
            -0.5 & -0.5 \\
            -0.5 & -0.5 \\
          \end{array}
        \right],
\end{eqnarray}
the value of the objective function in \eqref{eq:example1} is 15.5. Therefore, the solution to \eqref{eq:example2} and \eqref{eq:example1} is not necessarily a diagonal matrix. For other unitarily invariant matrix norms, we still do not know whether Lemma~\ref{lem1} holds or not.}
\else
\fi

\section{Conclusion}
{In this paper, we have investigated the capacity of a compound MIMO channel with an additive uncertainty of bounded spectral norm, and derived the optimal transmit covariance matrix in close-form.} When the channel uncertainty region is characterized by the spectral norm, we have revealed a hidden convexity property in this problem. We have proved that the optimal transmit covariance design is to diagonalize the nominal channel matrix and there is zero duality gap between the capacity of the compound MIMO Gaussian channel and the min-max channel capacity. 

\appendices
\section{Proof of Theorem \ref{thm1}}\label{app1}

First, we construct an upper bound of $C_{\max\min}$ by imposing {an} extra constraint in the inner minimization problem:
\begin{eqnarray}\label{eq:tmp_problem1}
&&\!\!\!\!\!\!\!\!\!\!\!~~~C_{\max\min}\nonumber\\
&&\!\!\!\!\!\!\!\!\!\!\!\overset{(a)}{\leq} \max_{\Qb\in\mathcal{Q}} \min_{\substack{\interleave\Hb-\Hb_0\interleave_2\leq \varepsilon\\
\Hb = \Ub_0 \bm\Sigma_\Hb \Vb_0^H}} \log\Det\left(\Ib_r + \gamma\Hb\Qb\Hb^H\right)\nonumber\\
&&\!\!\!\!\!\!\!\!\!\!\!\overset{(b)}{=}\max_{\Qb\in\mathcal{Q}} \min_{\substack{\interleave\Hb-\Hb_0\interleave_2\leq \varepsilon\\\Hb = \Ub_0 \bm\Sigma_\Hb \Vb_0^H}} \log\Det\left(\Ib_r + \gamma\bm\Sigma_\Hb\Vb_0^H\Qb\Vb_0\bm\Sigma_\Hb\right),\nonumber
\end{eqnarray}
where $\Ub_0$ and $\Vb_0$ are defined in \eqref{eq:SVD_H0}, step $(a)$ is due to the additional constraint in the inner minimization problem, and step $(b)$ is due to $\Hb = \Ub_0 \bm\Sigma_\Hb \Vb_0^H$ and $\Det(\Ib+\Ab\Bb)=\Det(\Ib+\Bb\Ab)$. Let us define $\tilde{\Qb}\triangleq \Vb_0^H\Qb\Vb_0$ and use $\Db(\tilde{\Qb})$ to denote the diagonal matrix that has the same diagonal elements {as $\tilde{\Qb}$. We then attain}
\begin{eqnarray}\label{eq:high}
&&\!\!\!\!\!\!\!\!\!\!\!~~~C_{\max\min}\nonumber\\
&&\!\!\!\!\!\!\!\!\!\!\!\overset{}{\leq}\max_{\Qb\in\mathcal{Q}} \min_{\substack{\interleave\Hb-\Hb_0\interleave_2\leq \varepsilon\\\Hb = \Ub_0 \bm\Sigma_\Hb \Vb_0^H}} \log\Det\left(\Ib_r + \gamma\bm\Sigma_\Hb\tilde{\Qb}\bm\Sigma_\Hb\right)\nonumber\\
&&\!\!\!\!\!\!\!\!\!\!\!\overset{(a)}{=} \max_{\tilde{\Qb}\in\mathcal{Q}} \min_{\substack{\interleave\Hb-\Hb_0\interleave_2\leq \varepsilon\\\Hb = \Ub_0 \bm\Sigma_\Hb \Vb_0^H}} \log\Det\left(\Ib_r + \gamma\bm\Sigma_\Hb \tilde{\Qb}\bm\Sigma_\Hb\right)\nonumber\\
&&\!\!\!\!\!\!\!\!\!\!\!\overset{(b)}{=} \max_{\tilde{\Qb}\in\mathcal{Q}} \min_{\interleave\bm\Sigma_\Hb -\bm\Sigma_{\Hb_0}\interleave_2\leq \varepsilon} \log\Det\left(\Ib_r + \gamma\bm\Sigma_\Hb \tilde{\Qb}\bm\Sigma_\Hb\right)\nonumber\\
&&\!\!\!\!\!\!\!\!\!\!\!\overset{(c)}{\leq} \max_{\tilde{\Qb}\in \mathcal{Q}} \min_{\interleave\bm\Sigma_\Hb-\bm\Sigma_{\Hb_0}\interleave_2\leq \varepsilon} \log\Det\left(\Ib_r + \gamma\bm\Sigma_\Hb \Db(\tilde{\Qb})\bm\Sigma_\Hb\right)\nonumber\\
&&\!\!\!\!\!\!\!\!\!\!\!\overset{(d)}{\leq} \max_{\Db(\tilde{\Qb})\in \mathcal{Q}} \min_{\interleave\bm\Sigma_\Hb-\bm\Sigma_{\Hb_0}\interleave_2\leq \varepsilon} \!\log\Det\left(\Ib_r + \gamma\bm\Sigma_\Hb \Db(\tilde{\Qb})\bm\Sigma_\Hb\right)\!\!\nonumber\\
&&\!\!\!\!\!\!\!\!\!\!\!\overset{(e)}{=} \max_{\substack{\bm \lambda\in B_\mathcal {Q}\\\bm \lambda\geq\bm0}}\min_{\substack{\|\bm \sigma-\bm \sigma_0\|_\infty\leq\varepsilon\\\bm \sigma\geq\bm0}}\sum_{i=1}^{\min\{t,r\}}\log (1+\gamma\sigma_i^2\lambda_i),
\end{eqnarray}
where step $(a)$ is due to \eqref{eq:power-region_property}, step $(b)$ is due to $\interleave\bm\Sigma_\Hb -\bm\Sigma_{\Hb_0}\interleave_2=\interleave\Hb-\Hb_0\interleave_2$ which is derived from $\Hb_0 = \Ub_0 \bm\Sigma_{\Hb_0} \Vb_0^H$, $\Hb = \Ub_0 \bm\Sigma_\Hb \Vb_0^H$ and \eqref{eq:channel_con_property}, step $(c)$ is due to the Hadamard inequality $\Det(\Ab)\leq \prod_{i} \Ab_{ii}$, step $(d)$ is because the feasible region $\Db(\tilde{\Qb})\in \mathcal{Q}$ is larger than the region $\tilde{\Qb}\in \mathcal{Q}$ according to \eqref{eq:power-region_property1}, and step $(e)$ is due to \eqref{eq:spectral_norm} and \eqref{eq:convex_set} with $\lambda_i$ representing the diagonal entries of $\Db(\tilde{\Qb})$.

Next, we build a lower bound of $C_{\max\min}$ by considering one extra constraint in the outer maximization problem:
\begin{eqnarray}\label{eq:tmp_problem2}
&&\!\!\!\!\!\!\!\!\!\!\!~~~C_{\max\min}\!\!\!\!\!\!\!\!\!\!\!\nonumber\\
&&\!\!\!\!\!\!\!\!\!\!\!\overset{(a)}{\geq} \max_{\substack{\Qb\in\mathcal{Q}\\\Qb=\Vb_0 \bm\Lambda_\Qb\Vb_0^H}} \min_{\Hb\in\mathcal{H}} I(\Qb,\Hb)\nonumber\\
&&\!\!\!\!\!\!\!\!\!\!\!=\!\max_{\substack{\Qb\in\mathcal{Q}\\\Qb=\Vb_0 \bm\Lambda_\Qb\Vb_0^H}}\! \min_{\interleave\bm\Delta\interleave_2\leq \varepsilon}\! \log\Det\left[\Ib_r\! + \! \gamma(\Hb_0\!+\!\bm\Delta)\Qb(\Hb_0\!+\!\bm\Delta)^H\right]\nonumber\\
&&\!\!\!\!\!\!\!\!\!\!\!\overset{(b)}{=}\!\max_{\substack{\bm\Lambda_\Qb\in\mathcal{Q}}}\! \min_{\interleave\tilde{\bm\Delta}\interleave_2\leq \varepsilon}\! \log\Det\left[\Ib_r\! + \! \gamma(\bm\Sigma_{\Hb_0}\!+\!\tilde{\bm\Delta})\bm\Lambda_\Qb
(\bm\Sigma_{\Hb_0}\!+\!\tilde{\bm\Delta})^H\right],\!\!\!\nonumber\\
\end{eqnarray}
where $\tilde{\bm\Delta}\triangleq\Ub_0^H\bm\Delta \Vb_0$, step $(a)$ is due the additional constraint in the outer maximization, and step $(b)$ is due to $\Hb_0 = \Ub_0 \bm\Sigma_{\Hb_0} \Vb_0^H$, $\Qb=\Vb_0 \bm\Lambda_\Qb\Vb_0^H$, the definition $\tilde{\bm\Delta}\triangleq\Ub_0^H\bm\Delta \Vb_0$, and the unitarily invariant properties \eqref{eq:channel_con_property} and \eqref{eq:power-region_property}. 

According to Lemma~\ref{lem1}, the optimal $\tilde{\bm\Delta}$ is a diagonal matrix. Hence, $\bm\Sigma_{\Hb}'\triangleq\bm\Sigma_{\Hb_0}\!+\!\tilde{\bm\Delta}$ in \eqref{eq:tmp_problem2} is also a diagonal matrix. Substituting this into \eqref{eq:tmp_problem2}, we have
\begin{eqnarray}\label{eq:low}
C_{\max\min}\!\!\!\!\!\!\!\!\!\!\!&&\geq\max_{\substack{\bm\Lambda_\Qb\in\mathcal{Q}}} \min_{\interleave \bm\Sigma_{\Hb}'-\bm\Sigma_{\Hb_0}\interleave_2\leq \varepsilon} \log\Det\left[\Ib_r\! + \! \gamma\bm\Sigma_{\Hb}'\bm\Lambda_\Qb
\bm\Sigma_{\Hb}'\right]\nonumber\\
&&=\max_{\substack{\bm \lambda\in B_\mathcal {Q}\\\bm \lambda\geq\bm0}}\min_{\substack{\|\bm \sigma-\bm \sigma_0\|_\infty\leq\varepsilon\\\bm \sigma\geq\bm0}}\sum_{i=1}^{\min\{t,r\}}\log (1+\gamma\sigma_i^2\lambda_i),
\end{eqnarray}
where the last step is due to \eqref{eq:spectral_norm} and \eqref{eq:convex_set} with $\sigma_i$ representing the diagonal entries of $\bm\Sigma_{\Hb}'$.
Using \eqref{eq:high} and \eqref{eq:low}, the optimal objective value of \eqref{eq:maxmin} is given by \eqref{eq:vec_opt_problem}.

Finally, we show \eqref{eq:solution} is an optimal solution to \eqref{eq:maxmin}. For this, we substitute the solution \eqref{eq:solution} into \eqref{eq:maxmin}, i.e.,
\begin{eqnarray}
&&\!\!\!\!\!\!\!\!\!\!\!~~~\max_{\substack{\Qb\in\mathcal{Q}\\\Qb=\Vb_0 \bm\Lambda_\Qb^\star\Vb_0^H}} \min_{\substack{\interleave\Hb-\Hb_0\interleave_2\leq \varepsilon\\\Hb = \Ub_0 \bm\Sigma_\Hb^\star \Vb_0^H}} \log\Det\left(\Ib_r + \gamma\Hb\Qb\Hb^H\right)\nonumber\\
&&\!\!\!\!\!\!\!\!\!\!\!\overset{}{=}\max_{\substack{\Qb\in\mathcal{Q}\\\Qb=\Vb_0 \bm\Lambda_\Qb^\star\Vb_0^H}} \min_{\substack{\interleave\Hb-\Hb_0\interleave_2\leq \varepsilon\\\Hb = \Ub_0 \bm\Sigma_\Hb^\star \Vb_0^H}} \log\Det\left(\Ib_r\! + \! \gamma\bm\Sigma_{\Hb}^\star\bm\Lambda_\Qb^\star
\bm\Sigma_{\Hb}^\star\right)  \nonumber\\
&&\!\!\!\!\!\!\!\!\!\!\!\overset{(a)}{=}\max_{\substack{\bm\Lambda_\Qb^\star\in\mathcal{Q}}} \min_{\interleave \bm\Sigma_{\Hb}^\star-\bm\Sigma_{\Hb_0}\interleave_2\leq \varepsilon} \log\Det\left(\Ib_r\! + \! \gamma\bm\Sigma_{\Hb}^\star\bm\Lambda_\Qb^\star
\bm\Sigma_{\Hb}^\star\right)  \nonumber\\
&&\!\!\!\!\!\!\!\!\!\!\!\overset{(b)}{=}\max_{\substack{\bm \lambda\in B_\mathcal {Q}\\\bm \lambda\geq\bm0}}\min_{\substack{\|\bm \sigma-\bm \sigma_0\|_\infty\leq\varepsilon\\\bm \sigma\geq\bm0}}\sum_{i=1}^{\min\{t,r\}}\log (1+\gamma\sigma_i^2\lambda_i)\nonumber\\
&&\!\!\!\!\!\!\!\!\!\!\!=C_{\max\min}, \label{eq:2}
\end{eqnarray}
where step $(a)$ is due to \eqref{eq:channel_con_property} and \eqref{eq:power-region_property}, step $(b)$ is due to \eqref{eq:convex_set} and \eqref{eq:vec_norm}. \hfill$\blacksquare$
\section{Proof of Theorem \ref{thm2}}\label{app2}

Consider the following upper bound of $C_{\min\max}$:
\begin{eqnarray}\label{eq:app2_1}
C_{\min\max}\!\!\!\!\!\!\!\!\!\!\!&&\leq \min_{\substack{\Hb\in \mathcal{H}\\\Hb = \Ub_0 \bm\Sigma_\Hb \Vb_0^H}}\max_{\Qb\in\mathcal{Q}} \log\Det\left(\Ib_r\! +\! \gamma\Hb\Qb\Hb^H\right)\nonumber\\
&&\overset{(a)}{=}\min_{\substack{\Hb\in \mathcal{H}\\\Hb = \Ub_0 \bm\Sigma_\Hb \Vb_0^H}}\max_{\substack{\Qb\in\mathcal{Q}\\\Qb=\Vb_0 \bm\Lambda_\Qb\Vb_0^H}} \! \log\Det\left(\Ib_r\! +\! \gamma\Hb\Qb\Hb^H\right)\nonumber\\
&&\overset{(b)}{=}\min_{\substack{\|\bm \sigma-\bm \sigma_0\|\leq\varepsilon\\\bm \sigma\geq\bm0}}\max_{\substack{\bm \lambda\in B_\mathcal {Q}\\\bm \lambda\geq\bm0}}\sum_{i=1}^{\min\{t,r\}}\log (1+\gamma\sigma_i^2\lambda_i),
\end{eqnarray}
where step $(a)$ is due to that the optimal power allocation result is of the form $\Qb=\Vb_0 \bm\Lambda_\Qb\Vb_0^H$ by using \eqref{eq:power-region_property}, \eqref{eq:power-region_property1}, and the Hadamard inequality \cite{Telatar1999}, step $(b)$ is derived by using \eqref{eq:channel_con_property}, \eqref{eq:power-region_property}, \eqref{eq:convex_set}, and \eqref{eq:vec_norm} as in \eqref{eq:2}.

Then, we construct a lower bound of $C_{\min\max}$:
\begin{eqnarray}\label{eq:app2_2}
C_{\min\max}\!\!\!\!\!\!\!\!\!\!\!&&= \min_{\interleave\Hb-\Hb_0\interleave\leq \varepsilon}\max_{\Qb\in\mathcal{Q}} \log\Det\left(\Ib_r\! +\! \gamma\Hb\Qb\Hb^H\right)\nonumber\\
&&\overset{(a)}{=}\min_{\interleave\Hb-\Hb_0\interleave\leq \varepsilon}\max_{\Qb\in\mathcal{Q}} \log\Det\left(\Ib_r\! +\! \gamma\bm\Sigma_{\Hb}\bm\Lambda_\Qb
\bm\Sigma_{\Hb}\right)
\nonumber\\
&&\overset{(b)}{=}\min_{\interleave\Hb-\Hb_0\interleave\leq \varepsilon}\max_{\Lambda_\Qb\in\mathcal{Q}} \log\Det\left(\Ib_r\! +\! \gamma\bm\Sigma_{\Hb}\bm\Lambda_\Qb
\bm\Sigma_{\Hb}\right)
\nonumber\\
&&\overset{(c)}{\geq}\min_{\interleave\bm\Sigma_{\Hb}-\bm\Sigma_{\Hb_0}\interleave\leq \varepsilon}\max_{\Lambda_\Qb\in\mathcal{Q}} \log\Det\left(\Ib_r\! +\! \gamma\bm\Sigma_{\Hb}\bm\Lambda_\Qb
\bm\Sigma_{\Hb}\right)
\nonumber\\
&&\overset{(d)}{=}\min_{\substack{\|\bm \sigma-\bm \sigma_0\|\leq\varepsilon\\\bm \sigma\geq\bm0}}\max_{\substack{\bm \lambda\in B_\mathcal {Q}\\\bm \lambda\geq\bm0}}\sum_{i=1}^{\min\{t,r\}}\log (1+\gamma\sigma_i^2\lambda_i),
\end{eqnarray}
where step $(a)$ is due to the optimal power allocation result by using \eqref{eq:power-region_property}, \eqref{eq:power-region_property1}, and the Hadamard inequality \cite{Telatar1999},
step $(b)$ is due to \eqref{eq:power-region_property}, step $(c)$ is due to the following result for unitarily invariant matrix norm: \cite[Theorem 7.4.51]{Hornbook1985} \cite[Eq. (3.5.30)]{Hornbook1991}
$$\interleave\bm\Sigma_{\Hb}-\bm\Sigma_{\Hb_0}\interleave\leq \interleave{\Hb}-{\Hb_0}\interleave,$$
with $\bm\Sigma_{\Hb}$ and $\bm\Sigma_{\Hb_0}$ being the diagonal matrices in the SVDs of $\Hb$ and ${\Hb_0}$, and step $(d)$ is due to \eqref{eq:convex_set} and \eqref{eq:vec_norm}. Then, \eqref{eq:vec_opt_problem1} follows from \eqref{eq:app2_1} and \eqref{eq:app2_2}. \eqref{eq:thm2_1} can be proved similarly to \eqref{eq:solution}. Hence, the theorem is proven. \hfill$\blacksquare$

\ifreport
\section{Proof of Lemma~\ref{lem1}}\label{app3}
In order to prove Lemma~\ref{lem1}, we first need to show the following result:
\begin{lemma}\label{lem:1}
Let $\bm\Sigma=[\varsigma_{ij}]\in\Real^{r\times t}$ be a diagonal matrix with non-negative diagonal entries $\varsigma_{ii}\geq0$ for $i=1,\cdots,\min\{r,t\}$, and $\Deltab\in\Cplx^{r\times t}$ be a matrix satisfying $\interleave{\bm\Delta}\interleave_2\leq \varepsilon$.
\begin{itemize}
\item[1)] The following inequality holds:
\begin{eqnarray}\label{eq:lem1-1}
\Det\left[(\bm\Sigma + \Deltab)^H(\bm\Sigma + \Deltab)\right]\geq\prod_{j=1}^{t}\max\{\varsigma_{jj}-\varepsilon,0\}^2,
\end{eqnarray}
where $\varsigma_{jj}$ are defined by $\varsigma_{jj}=0$ for $\min\{t,r\}< j\leq t$.

\item[2)] Let $S$ be a proper subset of $\{1,2,\cdots,t\}$, then
\begin{eqnarray}\label{eq:lem1-2}
&&\!\!\!\!\!\!\!\!\!\!\!~~~\Det\left\{\left[\bm\Sigma(S) + \Deltab(S)\right]^H\left[\bm\Sigma(S) + \Deltab(S)\right]\right\}\nonumber\\
&&\!\!\!\!\!\!\!\!\!\!\!\geq\prod_{j\in S }\max\{\varsigma_{jj}-\varepsilon,0\}^2,~~~ \forall S\subseteq \{1,2,\cdots,t\},
\end{eqnarray}
where $\Xb(S)$ denotes the submatrix of $\Xb$ obtained by deleting the rows and columns complementary to those indicated by $S$ from $\Xb$.

\item[3)] Let $\Ab=(\bm\Sigma + \Deltab)^H(\bm\Sigma + \Deltab)$, then:
\begin{eqnarray}\label{eq:lem1-3}
\Det\left[\Ab(S)\right]\geq\prod_{j\in S }\max\{\varsigma_{jj}-\varepsilon,0\}^2,\nonumber\\
~~~\forall S\subseteq \{1,2,\cdots,t\}.
\end{eqnarray}
\end{itemize}
\end{lemma}
\begin{proof}
1) It is known that for any $\Ab,\Bb\in\Cplx^{r\times t}$ the following singular value inequality holds \cite[Eq. (5.12.15)]{Mayerbook2000}, \cite[Corollary 7.3.8]{Hornbook1985}
\begin{eqnarray}
\left|\sigma_i(\Ab+\Bb)-\sigma_i(\Ab)\right|\leq \interleave\Bb\interleave_2,~~~\forall i=1,\cdots,\min\{t,r\}.\nonumber
\end{eqnarray}

By this, we have
\begin{eqnarray}\label{eq:singularbound}
&&\!\!\!\!\!\!\!\!~~~\sigma_i(\bm\Sigma + \Deltab)\nonumber\\
&&\!\!\!\!\!\!\!\!\geq \max\left\{\sigma_i(\bm\Sigma )-\interleave{\bm\Delta}\interleave_2,0\right\}\nonumber\\
&&\!\!\!\!\!\!\!\!
\geq\max\left\{\sigma_i(\bm\Sigma )-\varepsilon,0\right\},~~~\forall i=1,\cdots,\min\{t,r\},
\end{eqnarray}
where the maximization is due to the fact $\sigma_i(\bm\Sigma + \Deltab)\geq0$.
Moreover, $\sigma_i(\bm\Sigma + \Deltab)=0$ for $\min\{t,r\}<i\leq t$.

Since $\varsigma_{ii}\geq0$, the singular values of the diagonal matrix $\bm\Sigma$ are given by $\{\varsigma_{11},\varsigma_{22},\cdots,\varsigma_{pp},0,\cdots,0\}$. Let us define $\varsigma_{jj}$ as $\varsigma_{jj}=0$ for $\min\{t,r\}< j\leq t$.
Hence, we attain
\begin{eqnarray}
&&~~~\Det\left[(\bm\Sigma + \Deltab)^H(\bm\Sigma + \Deltab)\right]\nonumber\\
&&=\prod_{j=1}^t \sigma_j(\bm\Sigma + \Deltab)^2\nonumber\\
&&\geq \prod_{j=1}^t \max\left\{\sigma_j(\bm\Sigma )-\varepsilon,0\right\}^2\nonumber\\
&&=\prod_{j=1}^t\max\{\varsigma_{jj}-\varepsilon,0\}^2.\nonumber
\end{eqnarray}

2) Since $\bm\Sigma$ is a diagonal matrix, after deleting the rows and columns, the singular values of the submatrix $\bm\Sigma(S)$ are given by $\{\varsigma_{ii}:i\in S\}$. Moreover, after deleting some rows and columns, the spectral norm of $\Deltab(S)$ satisfies
$\interleave\Deltab(S)\interleave_2\leq\interleave{\bm\Delta}\interleave_2$ \cite[Thoerem 7.3.9]{Hornbook1985}.
Therefore
\begin{eqnarray}
&&~~~\Det\left\{\left[\bm\Sigma(S) + \Deltab(S)\right]^H\left[\bm\Sigma(S) + \Deltab(S)\right]\right\}\nonumber\\
&&=\prod_{j=1}^{|S|} \sigma_j[\bm\Sigma(S) + \Deltab(S)]^2\nonumber\\
&&\geq \prod_{j=1}^{|S|} \max\left\{\sigma_j(\bm\Sigma(S) )-\interleave\Deltab(S)\interleave_2,0\right\}^2\nonumber\\
&&\geq \prod_{j=1}^{|S|} \max\left\{\sigma_j(\bm\Sigma(S) )-\varepsilon,0\right\}^2\nonumber\\
&&=\prod_{j\in S}\max\{\varsigma_{jj}-\varepsilon,0\}^2.\nonumber
\end{eqnarray}

3) For any given $S\subseteq \{1,2,\cdots,t\}$, one can interchange the rows and columns of $\Ab$ and $(\bm\Sigma + \Deltab)$ by multiplying with
two permutation matrices $\Pb\in\Hmit_+^t$ and $\Qb\in\Hmit_+^r$ as
$\Bb=\Pb\Ab \Pb$ and $\Phib=\Qb(\bm\Sigma + \Deltab)\Pb$, such that
$\Ab(S)$ and $\bm\Sigma(S) + \Deltab(S)$ are the leading submatrices of $\Bb$ and $\Phib$, respectively, i.e.,
\begin{eqnarray}
&&\Bb = \left(
          \begin{array}{cc}
            \Ab(S) & \Cb \\
            \Cb^H & \Db \\
          \end{array}
        \right)
,~~\Phib = \left(
          \begin{array}{cc}
            \bm\Sigma(S) + \Deltab(S) & \Mb \\
            \Nb & \Gb \\
          \end{array}
        \right),\nonumber
\end{eqnarray}
Since $\Pb^H=\Pb$, $\Qb^H=\Qb$, $\Pb^2=\Ib$ and $\Qb^2=\Ib$, we attain
$$\Bb=\Pb(\bm\Sigma + \Deltab)^H\Qb\Qb(\bm\Sigma + \Deltab) \Pb=\Phib^H\Phib,$$
thereby
\begin{eqnarray}
\Ab(S) = \left(\bm\Sigma(S) + \Deltab(S)\right)^H\left(\bm\Sigma(S) + \Deltab(S)\right) + \Nb^H\Nb.\nonumber
\end{eqnarray}
Then,
\begin{eqnarray}
\Ab(S)\succeq \left(\bm\Sigma(S) + \Deltab(S)\right)^H\left(\bm\Sigma(S) + \Deltab(S)\right)\succeq\bm 0,\nonumber
\end{eqnarray}
which further implies that \cite[Corollary 7.7.4]{Hornbook1985}
\begin{eqnarray}
\Det\left[\Ab(S)\right]\geq \Det\left[\left(\bm\Sigma(S) + \Deltab(S)\right)^H\left(\bm\Sigma(S) + \Deltab(S)\right)\right].\nonumber
\end{eqnarray}
Finally, by using \eqref{eq:lem1-2}, the result in \eqref{eq:lem1-3} follows.
\end{proof}

Using part (3) of Lemma \ref{lem:1}, we can establish the following lemma:
\begin{lemma} \label{lem:key}
Let $\bm\Sigma=[\varsigma_{ij}]\in\Real^{r\times t}$ and $\Db=[d_{ij}]\in\Real^{t\times t}$ be two diagonal matrices with non-negative diagonal entries, and $\Deltab\in\Cplx^{r\times t}$ be a matrix satisfying $\interleave{\bm\Delta}\interleave_2\leq \varepsilon$. Then
\begin{eqnarray}\label{eq:lem:key1}
&&~~~\Det\left[\Ib_{t} + (\bm\Sigma + \Deltab)^H(\bm\Sigma + \Deltab)\Db\right]\nonumber\\
&&\geq\prod_{j=1}^{\min\{t,r\}}\left(1+\max\{\varsigma_{jj}-\varepsilon,0\}^2d_{jj}\right),
\end{eqnarray}
where equality is achieved by the diagonal matrix $\Deltab^\star$ with diagonal entries given by
\begin{eqnarray}
\Deltab^\star_{jj}=\varsigma_{jj}-\max\{\varsigma_{jj}\!-\!\varepsilon,0\},~j=1,\cdots,\min\{t,r\}.\nonumber
\end{eqnarray}
\end{lemma}
Then, Lemma~\ref{lem1} follows from \eqref{eq:lem:key1}.
\begin{proof}
We prove the inequality
\begin{eqnarray}\label{eq:lem:key}
&&~~~\Det\left[\Ib_{t} + (\bm\Sigma + \Deltab)^H(\bm\Sigma + \Deltab)\Db\right]\nonumber\\
&&\geq\prod_{j=1}^t\left(1+\max\{\varsigma_{jj}-\varepsilon,0\}^2d_{jj}\right)
\end{eqnarray}
by induction, where $\varsigma_{jj}$ for $\min\{t,r\}< j\leq t$ are defined as $\varsigma_{jj}=0$.

Let $\Bb=(\bm\Sigma + \Deltab)^H(\bm\Sigma + \Deltab)\Db$.
By part (3) of Lemma \ref{lem:1} and $\Db$ being a diagonal matrix, we attain
\begin{eqnarray}\label{eq:lem:key:proof0}
\Det\left[\Bb(S)\right]
\geq\prod_{j\in S }\left(\max\{\varsigma_{jj}\!-\!\varepsilon,0\}^2d_{jj}\right),\nonumber\\
~~~\forall S\subseteq \{1,2,\cdots,t\}.
\end{eqnarray}
For any set $S$ satisfying $\{1\}\subseteq S\subseteq\{1,2,\cdots,t\}$, by the cofactor (Laplace) expansion \cite[Eq. (6.2.5)]{Mayerbook2000} of $\Det[\Bb(S)]$ and \eqref{eq:lem:key:proof0}, we attain
\begin{eqnarray}\label{eq:lem:key:proof}
&&\!\!\!\!\!\!\!\!\!\!\!~~~\Det\left[\Eb_{11}+\Bb(S)\right]\nonumber\\
&&\!\!\!\!\!\!\!\!\!\!\!=\Det[\Bb(S)]+\Det\left[\Bb(S\setminus\{1\})\right]\nonumber\\
&&\!\!\!\!\!\!\!\!\!\!\!\geq\prod_{j\in S }\left(\max\{\varsigma_{jj}\!-\!\varepsilon,0\}^2d_{jj}\right)\!+\!\prod_{j\in S\setminus\{1\}}\!\left(\max\{\varsigma_{jj}\!-\!\varepsilon,0\}^2d_{jj}\right)\nonumber\\
&&\!\!\!\!\!\!\!\!\!\!\!= \left(1\!+\!\max\{\varsigma_{11}\!-\!\varepsilon,0\}^2d_{11}\right)\!\!\!\prod_{j\in S\setminus\{1\}}\!\!\!\!\left(\max\{\varsigma_{jj}\!-\!\varepsilon,0\}^2d_{jj}\right).
\end{eqnarray}

Similarly, for any set $S$ satisfying $\{1,2\}\subseteq S\subseteq\{1,2,\cdots,t\}$, by the cofactor expansion of $\Det[\Eb_{11}+\Bb(S)]$ and \eqref{eq:lem:key:proof}, we have
\begin{eqnarray}
&&\!\!\!\!\!\!\!\!\!\!\!~~~\Det\left[\Eb_{11}+\Eb_{22}+\Bb(S)\right]\nonumber\\
&&\!\!\!\!\!\!\!\!\!\!\!=\Det[\Eb_{11}+\Bb(S)]+\Det\left[\Eb_{11}+\Bb(S\setminus\{2\})\right]\nonumber\\
&&\!\!\!\!\!\!\!\!\!\!\!\geq\left(1\!+\!\max\{\varsigma_{11}\!-\!\varepsilon,0\}^2d_{11}\right)\!\!\!\prod_{j\in S\setminus\{1\}}\left(\max\{\varsigma_{jj}\!-\!\varepsilon,0\}^2d_{jj}\right)\nonumber\\
&&\!\!\!\!\!\!\!\!\!\!\!~~~+\left(1\!+\!\max\{\varsigma_{11}\!-\!\varepsilon,0\}^2d_{11}\right)\!\!\!\prod_{j\in S\setminus\{1,2\}}\!\!\!\left(\max\{\varsigma_{jj}\!-\!\varepsilon,0\}^2d_{jj}\right)\nonumber\\
&&\!\!\!\!\!\!\!\!\!\!\!= \prod_{j=1}^2\left(1\!+\!\max\{\varsigma_{jj}\!-\!\varepsilon,0\}^2d_{jj}\right)\!\!\!\prod_{j\in S\setminus\{1,2\}}\!\!\!\left(\max\{\varsigma_{jj}\!-\!\varepsilon,0\}^2d_{jj}\right).\nonumber
\end{eqnarray}

Suppose that for any set $S$ satisfying $\{1,2,\cdots,k\}\subseteq S\subseteq\{1,2,\cdots,t\}$, the following inequalities hold
\begin{eqnarray}\label{eq:lem:key:proof1}
&&~~~\Det\left[\Eb_{11}+\cdots+\Eb_{kk}+\Bb(S)\right]\nonumber\\
&&\geq \prod_{j=1}^k\left(1\!+\!\max\{\varsigma_{jj}\!-\!\varepsilon,0\}^2d_{jj}\right)\nonumber\\
&&~~~~\times\prod_{j\in S\setminus\{1,2,\cdots,k\}}\!\!\!\left(\max\{\varsigma_{jj}\!-\!\varepsilon,0\}^2d_{jj}\right).
\end{eqnarray}
Then, for any set $S$ satisfying $\{1,2,\cdots,k+1\}\subseteq S\subseteq\{1,2,\cdots,t\}$,
by the cofactor expansion of $\Det\left[\Eb_{11}+\cdots+\Eb_{kk}+\Bb(S)\right]$ and \eqref{eq:lem:key:proof1}, we have
\begin{eqnarray}
&&\!\!\!\!\!\!\!\!\!\!\!~~~\Det\left[\Eb_{11}+\cdots+\Eb_{(k+1)(k+1)}+\Bb(S)\right]\nonumber\\
&&\!\!\!\!\!\!\!\!\!\!\!=\Det[\Eb_{11}+\cdots+\Eb_{kk}+\Bb(S)]\nonumber\\
&&\!\!\!\!\!\!\!\!\!\!\!~~~+\Det\left[\Eb_{11}+\cdots+\Eb_{kk}+\Bb(S\setminus\{k+1\})\right]\nonumber\\
&&\!\!\!\!\!\!\!\!\!\!\!\geq\prod_{j=1}^k\left(1\!+\!\max\{\varsigma_{jj}\!-\!\varepsilon,0\}^2d_{jj}\right)\!\!\!\nonumber\\
&&~~~~~~\times\prod_{j\in S\setminus\{1,2,\cdots,k\}}\left(\max\{\varsigma_{jj}\!-\!\varepsilon,0\}^2d_{jj}\right)\nonumber\\
&&\!\!\!\!\!\!\!\!\!\!\!~~~+\prod_{j=1}^k\left(1\!+\!\max\{\varsigma_{jj}\!-\!\varepsilon,0\}^2d_{jj}\right)\nonumber\\
&&~~~~~~\times\prod_{j\in S\setminus\{1,2,\cdots,k+1\}}\!\!\!\left(\max\{\varsigma_{jj}\!-\!\varepsilon,0\}^2d_{jj}\right)\nonumber\\
&&\!\!\!\!\!\!\!\!\!\!\!= \prod_{j=1}^{k+1}\left(1\!+\!\max\{\varsigma_{jj}\!-\!\varepsilon,0\}^2d_{jj}\right)\nonumber\\
&&~~~~~~\times\prod_{j\in S\setminus\{1,2,\cdots,k+1\}}\!\!\!\left(\max\{\varsigma_{jj}\!-\!\varepsilon,0\}^2d_{jj}\right)\nonumber
\end{eqnarray}
By induction, the result of \eqref{eq:lem:key} follows.

Finally, \eqref{eq:lem:key} reduces to \eqref{eq:lem:key1} since $\varsigma_{jj}=0$ for $\min\{t,r\}< j\leq t$.
\end{proof}

\else
\fi
\bibliographystyle{IEEEtran}
\bibliography{refs_Compound}

\end{document}